\documentclass[10pt, conference, compsocconf]{IEEEtran}

\IEEEoverridecommandlockouts
\usepackage{cite}
\usepackage{amsmath,amssymb,amsfonts}
\usepackage{algorithmic}
\usepackage{graphicx}
\usepackage{textcomp}
\usepackage{xcolor}
\usepackage[ruled,linesnumbered]{algorithm2e}
\usepackage{multicol}
\usepackage{amsthm}
\newtheorem{theorem}{Theorem}[section]
\newtheorem{corollary}{Corollary}[theorem]
\newtheorem{lemma}[theorem]{Lemma}
\newtheorem{claim}[theorem]{Claim}
\newtheorem{definition}{Definition}[section]
\usepackage{comment}
\def\BibTeX{{\rm B\kern-.05em{\sc i\kern-.025em b}\kern-.08em
    T\kern-.1667em\lower.7ex\hbox{E}\kern-.125emX}}

\newcommand{\alessia}{}
\newcommand{\adnane}{}

\newcommand{\remove}[1]{}
\newtheorem{obs}{Observation}
\ifCLASSINFOpdf
\else
\fi
\begin{document}
\title{Upper and Lower Bounds for Deterministic Approximate Objects
}

\author{\IEEEauthorblockN{Danny Hendler}
\IEEEauthorblockA{\textit{Ben-Gurion University of the Negev}\\
Be'er Sheva, Israel \\
hendlerd@bgu.ac.il}
\and
\IEEEauthorblockN{Adnane Khattabi}
\IEEEauthorblockA{\textit{LaBRI, University of Bordeaux}\\
Bordeaux, France \\
adnane.khattabi-riffi@u-bordeaux.fr}
\and
\IEEEauthorblockN{Alessia Milani}
\IEEEauthorblockA{\textit{LaBRI, Bordeaux INP}\\
Bordeaux, France \\
milani@labri.fr}
\and
\IEEEauthorblockN{Corentin Travers}
\IEEEauthorblockA{\textit{LaBRI, Bordeaux INP}\\
Bordeaux, France \\
travers@labri.fr}
}


%


\maketitle

\begin{abstract}
Relaxing the sequential specification of shared objects has been proposed as a promising approach to obtain implementations with better complexity. In this paper, we study the step complexity of relaxed variants of two common shared objects: max registers and counters. In particular, we consider the $k$-multiplicative-accurate max register and the $k$-multiplicative-accurate counter, where read operations are allowed to err by a multiplicative factor of $k$ (for some $k \in  \mathbb{N}$). More accurately, reads are allowed to return an approximate value $x$ of the maximum value $v$ previously written to the max register, or of the number $v$ of increments previously applied to the counter, respectively, such that $v/k \leq x \leq v \cdot k$. We provide upper and lower bounds on the complexity of implementing these objects in a wait-free manner in the shared memory model.
\end{abstract}


%
\IEEEpeerreviewmaketitle

\section{Introduction}
With the ubiquitousness of multi-core and multi-processor systems, there is a growing need to gain better understanding of how to implement concurrent objects with improved complexity, while maintaining the natural correctness guarantee provided to programmers by linearizability. Relaxing the sequential specification of linearizable concurrent objects is one promising approach of achieving this \cite{AfekKY10,HenzingerKPSS13}. An object's \emph{sequential specification} defines its correct behavior in sequential executions. Roughly speaking, \emph{linearizability} \cite{herlihy1990linearizability} guarantees that any concurrent execution is equivalent to a sequential one. 

There is empirical evidence that relaxing the sequential specification of some common objects, e.g. queues and counters, yields improved performance of linearizable implementations, e.g \cite{HenzingerKPSS13,RukundoAT19}. However, the theoretical principles to implement concurrent objects more efficiently by relaxing their sequential specification are not yet clear.

In this paper, we study relaxed-semantics variants of two well-known concurrent objects -- counters and max registers, in the classical shared memory model. In particular, we investigate the extent to which allowing wait-free linearizable implementations of these objects to return approximate values, rather than accurate ones, may improve their step complexity.

A counter is a linearizable object that supports a $\mathit{CounterIncrement}$ operation and a $\mathit{CounterRead}$ operation. The sequential specification of a counter requires that a $\mathit{CounterRead}$ operation returns the number of $\mathit{CounterIncrement}$ operations that precede it. A relaxed variant of the counter is the \emph{k-multiplicative-accurate} counter, defined in \cite{AspnesCAH16}, where a $\mathit{CounterRead}$ operation returns an approximate value $x$ of the number $v$ of $\mathit{CounterIncrement}$ operations that precede it, such that $v/k \leq x \leq v \cdot k$ for some parameter $k > 0$.

A max register $r$ supports a $Write(v)$ operation that writes a non-negative integer $v$ to $r$ and a $Read$ operation that returns the maximum value previously written to $r$, \cite{AspnesCAH16}. We define the $k$-multiplicative-accurate max register by allowing a $Read$ operation to return an approximate value $x$ of the largest value $v$ written before it, such that $v/k \leq x \leq v \cdot k$ for some parameter $k > 0$.

\subsection{Related Work}

A well-known result by Jayanti, Tan and Toueg \cite{JayantiKT2000} proved a linear lower bound in the number of processes $n$ on the worst-case step complexity of obstruction-free implementations from historyless primitives (e.g. read,write, test\&set) of a large class of shared objects that includes exact counters.  A wait-free exact counter with optimal worst case step complexity can be constructed easily by using a \emph{wait-free atomic snapshot} : to increment the counter, a process simply increments its component of the snapshot, and to read the counter's value, it invokes \texttt{Scan} and returns the sum of all components in the view it obtains. Since wait-free atomic snapshot can be implemented, using reads and writes only, with worst-case step complexity linear in $n$, e.g. \cite{attiya2001adaptive}, so can counters. 

Aspnes, Attiya and Censor-Hillel \cite{AspnesAC2012} show the possibility of implementing exact counting algorithms whose step complexity is sub-linear when the number of operations is bounded. In particular, they presented a wait-free exact counter for which the step complexities of $\mathit{CounterIncrement}$  and $\mathit{CounterRead}$ operations are $O\big(min(\log n \log v, n)\big)$ and $O\big(min(\log v, n)\big)$, respectively, where $v$ is the object's current value. 
However, for executions in which the number of $\mathit{CounterIncrement}$ operations is exponential in $n$, both the worst-case and the amortized step complexities of their constuction become linear in $n$.

In \cite{BaigHMT19} Baig et al. present the first wait-free read/write exact counter whose amortized step complexity is polylogarithmic in $n$, $O(\log^2n)$, in executions of arbitrary length and prove that their algorithm is optimal in terms of amortized step complexity up to a logarithmic factor.

Approximate counting has many applications (e.g.;\cite{DiceLM13,AfekKKMT12}) and there is a large literature in approximate probabilistic counting, both in the sequential (e.g.\cite{Morris78a,Flajolet85}) and concurrent setting (e.g.\cite{AspnesC10,BenderG11}). In \cite{AspnesC10}, Aspnes and Censor present a randomized algorithm to implement an approximate counter that requires sublinear step complexity, and where any $\mathit{CounterRead}$ operation has a \emph{high probability} of returning a value which is at most a fraction of $\delta$ less than the number of increments that have finished before the read started, and at most a fraction of $\delta$ more than the number of increments that have started before the read finished. Similarly to their counter, the $k$-multiplicative accurate counter we study in this paper, allows a multiplicative error of the exact value. However, we study approximate \emph{deterministic counting}, where the value returned by $\mathit{CounterRead}$ operations is always ensured to be within the given approximation range.

Aspnes et al. also study $k$-additive counters that allow some additive error for the value returned by $\mathit{CounterRead}$ operations, \cite{AspnesAC2012}. In particular they prove a lower bound of $\Omega(min(n-1,\log{m} - \log{k}))$ on the worst-case step complexity of any deterministic asynchronous linearizable implementation of a $k$-additive counter, where $m$ is the number of states of the counter and $n$ is the number of processes. No matching upper bound is given.

The exact max register object has been proposed in \cite{AspnesAC2012} where Aspnes et al. present a bounded variant of this object to beat the linear lower bound on the worst case step complexity by Jayanti et al. In particular, their algorithm has $O(\log m)$ worst case step complexity for both $\mathit{Read()}$ and $\mathit{Write(v)}$ operations provided that the value written in the max register does not exceed the value $m$.  Their algorithm is optimal \cite{AspnesCAH16}. Baig et al. \cite{BaigHMT19} presented an unbounded deterministic max register implementation with polylogarithmic \textit{amortized} step complexity in executions of arbitrary length. 



\subsection{Our Contribution}
 \subsubsection*{k-multiplicative-accurate counter}
%
%
To the best of our knowledge we present the first deterministic approximate counter with \emph{constant amortized complexity}. More precisely, we present a wait-free linearizable $k$-multiplicative-accurate counter for $k\geq \sqrt{n}$ where $n$ is the number of processes, with \emph{constant} amortized step complexity for executions of arbitrary length. Then, by extension of the lower bound of Attiya and Hendler, \cite{AttiyaH10}, we prove that any $n$-process solo-terminating implementation of a $k$-multiplicative-accurate counter from read/write and conditional primitive operations (including $k$-word compare-and-swap) has amortized step complexity of $\Omega(\log(n/k^2))$, for $k \leq \sqrt{n}/2$. 
Our results together with the upper and lower bound on exact counting proved in \cite{BaigHMT19} show that when the approximation parameter $k$ does not depend on $n$, relaxing the counter semantics by allowing a multiplicative error cannot asymptotically reduce the amortized step complexity by more than a logarithmic factor. 

We also prove that the \emph{worst-case step complexity} of obstruction-free implementations of $m$-bounded $k$-multiplicative-accurate counters from \emph{historyless} primitives is $\Omega(\min(n,\log_2{\log_k{m}}))$, where $n$ is the number of processes and $m$ is a bound on the number of $\mathit{CounterIncrement}$ operation instances that can be performed on the counter. This implies that for unbounded k-multiplicative-accurate counters, the worst case step complexity is in $\Omega(n)$, and we fall back to the linear lower bound by Jayanti, Tan and Toueg  \cite{JayantiKT2000}. 

\subsubsection*{$k$-multiplicative-accurate max register}
We prove that relaxing the semantics of the bounded max register by allowing inaccuracy of even a constant multiplicative factor yields an exponential improvement in the worst-case step complexity. In particular, we prove that the worst-case step complexity of obstruction-free read/write implementations of $m$-bounded $k$-multiplicative-accurate max registers is $\Omega(\min(n,\log_2{\log_k{m}}))$, where $n$ is the number of processes. A max register is $m$-bounded, if it can only represent values in $\{0,\ldots,m-1\}$. Then, we present a novel $m$-bounded $k$-multiplicative-accurate max register algorithm whose worst-case step complexity matches this lower bound. 
We can easily ``plug-in'' our bounded k-multiplicative-accurate max register into the construction proposed by Baig et al. \cite{BaigHMT19} to obtain an unbounded k-multiplicative-accurate max register with sub-logarithmic amortized step complexity \adnane{(omitted due to space constraints)}. 

\section{Model and Preliminaries}
\label{sec:preliminaries}

We consider an asynchronous shared memory system, where a set $\mathcal{P}$ of $n$ crash-prone processes communicate by applying operations to
shared objects. An \emph{object} is an instance of an abstract
data type. It is characterized by a domain of possible values and
by a set of \emph{operations} that provide the only means to
manipulate it.  An \emph{implementation} of a shared object provides a
specific data-representation for the object from a set
of shared \emph{base objects}, each of which is assigned an
initial value; the implementation also provides algorithms for each process
in $\mathcal{P}$ to apply each operation to the object being implemented.
To avoid confusion, we call operations on the base objects
\emph{primitives} and reserve the term \emph{operations} for the
objects being implemented.

An \emph{execution fragment} is a (finite or infinite) sequence of steps performed by processes as they follow their algorithms. In each step, a process applies at most a single primitive to a base object (possibly in addition to some local computation). We consider read, write and test\&set primitives. An \emph{execution} is an execution fragment that starts from the \emph{initial configuration}. This is a configuration in which all base objects have their initial values and all processes are in their initial states.

A set of primitives is \emph{historyless} if all the nontrivial primitives in the set overwrite each other; we also require that each such primitive overwrites itself. A primitive is \emph{nontrivial} if it may change the value of the base object to which it is applied.


Operation $op_1$ \emph{precedes} operation $op_2$ in an execution $E$, if $op_1$'s response appears in $E$ before $op_2$'s invocation. We consider only \emph{deterministic implementations}, in which the next step taken by a process depends only on its state and the response it receives from the event it applies.

Roughly speaking, an implementation is \emph{linearizable} \cite{herlihy1990linearizability} if each operation appears to take effect atomically at some point between its invocation and response; it is \emph{wait-free} \cite{herlihy1991wait} if each process completes its operation if it performs a sufficiently large number of steps; it is \emph{obstruction-free} (also called \emph{solo-terminating}) \cite{herlihy2003obstruction} if each process completes its operation if it performs a sufficiently large number of steps when running solo.

The worst-case \emph{amortized step complexity} (henceforth simply amortized step complexity) is defined as the worst-case (taken over all possible executions) average number of steps performed by operations. It measures the performance of an implementation as a whole rather than the performances of individual operations. More precisely, given a finite execution $E$,  an operation  $Op$ \emph{appears} in $E$ if it is invoked in $E$. We denote by $\mathit{Nsteps}$$(op,E)$ the number of steps performed by $op$ in $E$ and by $Ops(E)$ the set of operations that appear in $E$. The amortized step complexity of an implementation $A$ is then:
\begin{equation*}
\mathit{AmtSteps}(A) = \max_{E} \frac{\sum_{op \in Ops(E)} \mathit{Nsteps}(op,E)}{|Ops(E)|}
\end{equation*}

\section{Unbounded Approximate $k$-multiplicative-accurate Counter}
We present a wait-free linearizable unbounded  $k$-multiplicative-accurate  counter with $k\geq\sqrt{n}$ whose amortized step complexity is constant (Algorithm~\ref{algo:unboundedcounter}).  

The algorithm uses an unbounded sequence of bits initially equal to $0$, denoted $\mathit{switch}_0,\mathit{switch}_1,\ldots$ to approximately keep track of the number of increments that have been performed by the processes.  For each $i\geq0$,  $\mathit{switch}_i$ can be accessed by $test\&set$ and $read$ operations. $\mathit{switch}_i.test\&set()$ sets the value of $\mathit{switch}_i$ to $1$ and returns its previous value. A $read$ simply returns the value of $\mathit{switch}_i$. 

\alessia{In a nutshell, each process locally keeps an accurate count of the number of $\mathit{CounterIncrement}$ operations it performs and that are not yet known by the other processes. When this count reaches a certain threshold, the process tries to inform other processes of the number of increments it has performed locally, by attempting to set to $1$ a switch in an appropriate bounded range. When a process succeeds in setting a switch to $1$, it will restart the local count from $0$.
$\mathit{switch}$ bits are set in increasing order with regards to their index, one after the other.}

\alessia{ In particular, the initial value of the threshold is 1 and after their first call to $\mathit{CounterIncrement}$, each process will attempt to set $\mathit{switch}_{0}$. Afterwards, the sequence of $\mathit{switch}_i$ with $i\geq 1$ is partitioned into consecutive intervals of size $k$. For any such interval $[qk+1,(q+1)k]$, where $k$ is an integer, and for any $j \in [qk+1,(q+1)k]$, $\mathit{switch}_j$ equals to $1$ indicates that  $k^{q+1}$  instances of $\mathit{CounterIncrement}$ have been performed by some process. In other words, a process $p$ locally performs $k^{q+1}$ instances of $\mathit{CounterIncrement}$ before attempting to set a switch in the interval $[qk+1,(q+1)k]$ and it increments its local threshold only if it knows that the last switch in this interval is set to $1$ (i.e.; at least $k\cdot k^{q+1}$ instances of $\mathit{CounterIncrement}$ have been performed). The threshold is incremented by a factor $k$. 
There is no guarantee that $p$ will succeed in setting to $1$ one of the switches. But in this case, sufficiently many increments have been performed by the processes so that a $\mathit{CounterRead}$ operation can safely ignore the  increments kept locally by $p_i$ and still returns a value within a bounded factor of the actual number of increments. }

\alessia{ 
By using test\&set to modify a $\mathit{switch}$ from $0$ to $1$, we ensure that the  $\mathit{CounterIncrement}$ instances accounted for by $\mathit{switch}_j$ are distinct from those accounted for by $\mathit{switch}_{j'}$, for any $j' \neq j$.
}


\alessia{Performing  an instance of a $\mathit{CounterRead}$ operation $op$ consists in traversing the sequence of switches until $0$ is found. An approximation of the total number of $\mathit{CounterIncrement}$ is then deduced from the index of the last switch $j$ that $op$ found equal to $1$. The value returned is the sum of the $\mathit{CounterIncrement}$ operations represented by each switch from $\mathit{switch}_0$ to $\mathit{switch}_j$. In particular, $\mathit{switch}_0$ counts for one $\mathit{CounterIncrement}$, and each $\mathit{switch}_i$ in an interval $[qk+1,(q+1)k]$ for some integer $q\leq q_j$ counts for $k^{q+1}$ $\mathit{CounterIncrement}$ operations, where $\mathit{switch}_j$ belongs to the interval $[q_jk+1,(q_j+1)k]$.
}
\subsection{The $\mathit{CounterIncrement}$ operation}
\adnane{
Each process $i$ is equipped with two persistent local variables, $\mathtt{lcounter}_i$ and $\mathtt{limit}_i$. The former stores the number of $\mathit{CounterIncrement}$ instances performed by process $i$ not yet announced to the other processes; and the latter stores the treshold on the number of $\mathit{CounterIncrement}$ that can be performed by process $i$ without informing the other processes.
}

When a $\mathit{CounterIncrement}$ operation is invoked by a process $i$, $\mathtt{lcounter}_i$ is first incremented (line~\ref{INC:incrementlcounter}). To ensure that a $\mathit{CounterRead}$ operation instance returns a value that is within a multiplicative factor $k$ of the actual number of increments, when $\mathtt{lcounter}_i$ reaches a certain threshold stored in $\mathtt{limit}_i$, process $i$ tries to inform the other processes of the number of increments  it has performed locally (lines~\ref{INC:condition-old}). The value of $\mathtt{limit}_i$ is initially $1$ and is multiplied by $k$ each time it is modified (line~\ref{INC:limit1} and line~\ref{INC:limit2}). When $\mathtt{lcount}_i = \mathtt{limit}_i =k^{q+1}$ for some integer $q$, process $i$ tries to set to $1$ one of the $k$ $\mathit{switch}_j$ whose index $j$ is in the corresponding range $[qk+1,(q+1)k]$ (lines~\ref{INC:forloop}- \ref{INC:return1}). If it succeeds, it resets the local counter $\mathtt{lcounter}_i$. The number of $\mathit{CounterIncrement}$ instances it has performed locally has been announced to the other processes, and thus will be taken  into account by future $\mathit{CounterRead}$ operations. 

Additionally, process $i$  writes the index of the switch it sets together with a sequence number into a shared variable $H[i]$ (lines \ref{INC:incrementsn} and  \ref{INC:writepair}). As explained later this  pair is intended to help $\mathit{CounterRead}$ operation instances to complete. Finally, the process will also update the value of the local persistant variable $l_0$ to indicate the index of the switch it managed to set within the interval (line \ref{INC:setl0}). By doing so, we ensure that the process will avoid attempting to reset the same switches every time it reaches the treshold of $limit_i$ in the current interval by starting from the index $qk + l_0$ in the next attempt. 
If it does not succeed, every $\mathit{switch}_j$, where $j \in [qk+1, (q+1)k]$ is set. 
We show in the proof that for $k \geq \sqrt{n}$, this number is sufficiently large for allowing $\mathit{CounterRead}$ operations to return values within a factor $k$ of the total number of $\mathit{CounterIncrement}$  instances (Section \ref{sec:proof-counter}).
The threshold $\mathtt{limit}_i$ is then incremented by a factor $k$ (line~\ref{INC:limit2}) and the value of $l_0$ is reset to $1$ (line~\ref{INC:resetl0}).

\begin{algorithm*}[htb!]
  \SetAlgoLined
  \DontPrintSemicolon
  \SetAlgoNoEnd
  \SetKwProg{Fn}{Function}{}{end}
  \textbf{Shared variables}\;
  $~~$ $\mathit{switch}_j \in \{0,1\} :$ for each $j \in \mathbb{N}$, a $1$-bit register that supports $test\&set$ and $read$ primitives, initially all $0$\;
  $~~$   $\mathit{H}[n] :$ an array of $n$ integer pairs $(val, sn)$\; 
  \BlankLine
  \textbf{Persistent local variables} \;
  $~~$ $\texttt{last}_i \in \mathbb{N}_0 :$  largest index of a switch  accessed by $i$,  initially  $0$\;
  $~~$   $\texttt{lcounter}_i$: number of unannounced \textit{CounterIncrement} by process $i$, initially $0$ \;
  $~~$   $\texttt{limit}_i$ :  number of  \textit{CounterIncrement} that process $i$ can perform locally, initially $1$ \;
  $~~$   $\texttt{sn}_i$:  number of switches set to $1$  by process $i$, initially $0$\;
  $~~$   $\texttt{l}_0$:  index of last switch accessed by the process $i$ in the current set of switches, initially $1$\;
  \BlankLine
  \begin{multicols}{2}
    \Fn{CounterIncrement()}{
      $\mathtt{lcounter}_i \gets  \mathtt{lcounter}_i + 1$     \label{INC:incrementlcounter}\;
      \If{$\mathtt{lcounter}_i = \mathtt{limit}_i$}{\label{INC:condition-old}
        $j \gets  \log_k(\mathtt{lcounter}_i) $ \label{INC:lcounter}\;
		\If{$j > 0$}{\label{INC:morethanone}\	
    \For{$ \ell \gets (j-1)k+ l_0,\ldots, j\cdot k$}{\label{INC:forloop}
      \If{$\mathtt{switch}_{\ell}.test\&set() = 0$}{\label{INC:setswitch}
        $\mathtt{sn}_i \gets \mathtt{sn}_i +1$\; \label{INC:incrementsn}
        $\mathit{H}[i] \gets (\ell,sn_i) $ \; \label{INC:writepair}
        $\mathtt{lcounter}_i \gets 0$\; \label{INC:resetCounter}
        \If{ $\ell = jk$ }{
		 $\mathtt{limit}_i \gets k\cdot{}\mathtt{limit}_i$ \label{INC:limit1} \;
}
	   $ l_0 \gets 1 + \ell \mod k$ \label{INC:setl0}\;
        \textbf{return} \label{INC:return1}\;
      }
    }
    $ l_0 \gets 1$  \label{INC:resetl0}\;    
}

\Else{
\If{$\mathtt{switch}_{0}.test\&set() = 0$}{
        $\mathtt{lcounter}_i \gets 0$\label{INC:setswitch0}\; 
}
      }

    $\mathtt{limit}_i \gets k \cdot \mathtt{limit}_i$ \label{INC:limit2}\;
    }
    \textbf{return}\;
  }
  {}
  \Fn{ReturnValue(p,q)}{ \label{ReturnValue}
	$ret \gets 1 + p\cdot k^{q+1}$ \label{READ:normalvalue} \label{READ:firstswitch}\;
	\uIf{$q \geq 1$}{
	$ret \gets  ret +\sum_{l=1}^{q} k^{l+1}$ \label{READ:previoussets} \;
	}
	\textbf{return} $k\cdot ret$ \label{READ:ktimesret}\;
\vspace{3mm}
}{}
  \Fn{CounterRead()}{
    $c \gets 0$\;
    \While{$switch_{\mathtt{last}_i} \neq 0$}{\label{READ:whilecondition}

	$p \gets last_i \mod k$\;
	$q \gets \lfloor \frac{last_i} {k}\rfloor$ \label{READ:normalq}\;
	\uIf{$last_i \text{ mod } k = 0$}{
      $\mathtt{last}_i \gets \mathtt{last}_i +1$ \label{READ:incrementlast}\; 
	}
	\Else{
		$\mathtt{last}_i \gets \mathtt{last}_i + k -1$ \label{READ:incrementlast2}\;
	} 

	$c \gets c+1$\;
      \If{ $c  \mod  n = 0$}{\label{READ:conditionhelping}
        \uIf{ $c = n $}{
          \For{$j \gets 1,\ldots, n$}{\label{forloop1}
            $\mathtt{help}_i[j] \gets H[j].sn$ \label{READ:firstsn}\;
          }
        }
        \Else{
          \For{$j\gets 1,\ldots,n$}{\label{forloop2}
            $(val,sn) \gets H[j]$\label{READ:secondsn}\;
            \If{ $sn - \mathtt{help}_i[j] \geq 2$}{ \label{READ:help}

			$p \gets val \mod k$ \; 
              $q \gets \lfloor \frac{val} {k}\rfloor$\label{READ:helpingvalue}\;
			\textbf{return} $\mathit{ReturnValue(p,q)}$ \label{READ:helpingreturn}\;

            }
          }
        }
      }
    }

	\uIf{$last_i = 0$}{\textbf{return} 0}
	\textbf{return} $\mathit{ReturnValue(p,q)}$ \label{READ:normalreturn}\;

}{}
\end{multicols}
\caption{$k$-multiplicative-accurate unbounded counter, code for process $i$.}
\label{algo:unboundedcounter}
\end{algorithm*}

\subsection{The $\mathit{CounterRead}$ operation}
\adnane{
When a $\mathit{CounterRead}$ operation is invoked, process $i$ scans the first and last $\mathit{switch}$ of each interval of $k$ switches, looking for the first one that is not yet set to $1$. When such a switch is found, the index $h$ of the last switch read that was equal to $1$ is stored in the persistent local variable $\mathtt{last}_i$ to avoid scanning the sequence from the beginning each time. We compute the value $ret$ returned by the $\mathit{CounterRead}$ operation in the function $\mathit{ReturnValue(p,q)}$ where $h=q\cdot k + p$ (line \ref{ReturnValue}). First, we consider the required increments needed to set all the switches in the current interval $[qk+1,(q+1)k]$ by adding to $ret$ the value $p\cdot k^{q+1}$ (line \ref{READ:normalvalue}). Next, we add $1$ to $ret$ to account for the first $\mathit{switch_0}$ (line \ref{READ:firstswitch}), and then for each previous interval $[(l-1)k+1,lk]$ where $1 \leq l \leq q$, we add $k^{l+1}$ to $ret$ (line \ref{READ:previoussets}). Finally, we return the computed value $ret$ multiplied by a factor $k$ to ensure $ret$ falls in the approximation range of the $k$-multiplicative-accurate counter. 
}

However, it may be the case  that the condition at line \ref{READ:whilecondition} is never verified, as other processes may concurrently keep executing $\mathit{CounterIncrement}$ operations. Thus, to ensure wait-freedom, we employ the following helping mechanism : a $\mathit{CounterIncrement}$ operation by a process $i$ that succeeds to set a $\mathit{switch}_j$, writes the index $j$ of this switch together with a sequence number in the shared register $H[i]$ (lines \ref{INC:incrementsn} and \ref{INC:writepair}). A $\mathit{CounterRead}$ operation $op$ that fails to find a switch to $0$ after $\theta(n)$ steps, reads all the $n$ shared registers $H[i]$ with $i \in {1,\ldots, n}$. If a consistent value is found, then it returns at line \ref{READ:helpingreturn}. Otherwise, it executes another $\theta(n)$ steps. The first time $op$ scan the array $H$, it stores the sequence number read in each $H[j]$, denotes sn$_j$. When scanning $H$ again, $op$ will select a pair whose timestamp is greater than or equal to sn$_j$ +2. This ensures, that the corresponding switch has been set by process $j$ in the execution interval of $op$.

\subsection{Proof}
\label{sec:proof-counter}
\subsubsection{Wait-freedom and technical lemmas}
Let $E$ be an execution of the $k$-multiplicative-accurate unbounded counter implemented in Algorithms \ref{algo:unboundedcounter}. 
\begin{lemma}\label{lemma:unboundedcounter-waitfreedom} Operations $\mathit{CounterIncrement()}$  and $\mathit{CounterRead()}$ are wait-free.
 \end{lemma}
 \begin{proof}
Let $op_r$ and $op_w$ denote a $\mathit{CounterRead}$ and $\mathit{CounterIncrement}$ instance respectively in $E$.  The number of steps taken during $op_w$ is bounded since at most the process will attempt to set $k$ switches during a call to  $\mathit{CounterIncrement}$ and there are no other loops or function calls in the $\mathit{CounterIncrement}$ operation.


Suppose by contradiction that $op_r$ does not terminate. Meaning that every bit $\mathit{switch}_\ell$ it reads has been set to $1$. Since the bits are initially $0$, there is at least one process $q$ that infinitely often performs  a successful  test\&set operation on these bits. Note that each time this occurs, $q$ increments its sequence number $\mathtt{sn}_q$ and reports the new value in the helping array $H$ (lines~\ref{INC:incrementsn}-~\ref{INC:writepair}). As every $n$ iterations of the \textbf{while} loop, $op_r$ scans the array $H$, it will eventually detects that the sequence number of $q$ has been incremented at least twice, hence $op_r$ terminates via the helping mechanism (lines~\ref{forloop2}-\ref{READ:helpingreturn}). 
Therefore, operations $\mathit{CounterIncrement}$  and $\mathit{CounterRead}$ are  wait-free.
 \end{proof}

We continue with a few technical lemmas.

\begin{lemma}\label{claim:incrementalSwitch} Switches are set to 1 in $E$ in increasing order of their index, starting from $\emph{switch}_0$.
\end{lemma}
\begin{proof} 
For each process $p$ the initial value of $limit_p$ is $1$ and of $counter_p$ is $0$, thus the first $\mathit{CounterIncrement}$ operation by process $p$ applies a test\&set primitive to $switch_0$ according to lines  \ref{INC:incrementlcounter}, \ref{INC:condition-old}, \ref{INC:lcounter}, and \ref{INC:setswitch0}. 
We now prove that for any given process $p$ and for any $j\geq 1$, $p$ applies a test\&set primitive (if any) on each of the switches with indexes in the interval $[(j-1)\cdot k +1,\ldots j \cdot k]$ in an increasing order of their index, starting from $switch_{(j-1)\cdot k +1}$. 
First observe that for any process $p$, the initial value of $l_0$ is $1$, and $l_0$ is set to $1$ iff the value of $limit_p$ is multiplied by a factor $k$ (lines \ref{INC:resetl0},\ref{INC:limit2} and lines \ref{INC:limit1},\ref{INC:setl0}). This implies that when a new $j$ is computed at line \ref{INC:lcounter}, the value of $l_0$ is $1$.

Then the first iteration of the \textbf{for loop} at line \ref{INC:forloop} starts at $l=(j-1)\cdot k+1$. Also, the value of $l$ is incremented by one at each iteration of the for loop at line \ref{INC:forloop} unless $p$ succesfully sets a $switch_{(j-1)k +i}$ with $i \in\{1,\ldots,k\}$. In this latter case, the value of $l_0$ is modified at line \ref{INC:setl0} and takes the value $i+1$ if $i<k$, or $1$ otherwise (we reach the end of the set). If $l_0$ takes a value different from $1$, that is $l \neq j\cdot k$, (otherwise the claim is proved), then the $\mathit{CounterIncrement}$ operation returns at line \ref{INC:return1} without modifying the value of $limit_p$. Thus, in the execution of a successive $\mathit{CounterIncrement}$ operation (if any), process $p$ will apply the next test\&set primitive (if any) to $switch_{jk+i+1}$ (because of lines \ref{INC:condition-old}, \ref{INC:lcounter}, \ref{INC:forloop}). 

The value of $limit_i$ is multiplied by $k$ (and than the value of $j$ is incremented by one) only after a process has applied a test\&set primitive (both successfully or not) to the last switch in the current interval $[(j-1)\cdot k +1, \ldots ,j\cdot k]$ with $log_k(limit_i)=j$ (lines \ref{INC:limit1}, \ref{INC:limit2}). This completes the proof. 

\end{proof}

 \begin{lemma}\label{lemma:HelpingRead} For any given execution $E$, if a $\mathit{CounterRead}$ operation $op$ returns the value computed in $\mathit{ReturnValue(p,q)}$ at line \ref{READ:helpingreturn}, then $switch_{q\cdot k +p}$ was equal to $0$ before the invocation of $op$ and the test\&set primitive that sets $switch_{q\cdot k +p}$ to $1$ is applied during the execution interval of $op$.
 \end{lemma}
 \begin{proof} At line \ref{READ:secondsn},  $op$ reads a pair $(val,\sigma)$ from an entry $H[p']$ of the helping array $H$ where $val = q\cdot k +p$. According to lines \ref{INC:setswitch}, \ref{INC:incrementsn}, and \ref{INC:writepair}, a unique process $p'$ sets to 1 the $switch_{val}$ and associates with $val$ the  sequence number $\sigma$ computed at line~\ref{INC:incrementsn}, before writing the pair $(v,\sigma)$ to $H[p']$ in the execution of a $\mathit{CounterIncrement}$ operation $op'$.

Let $p$ be the process that executes the $\mathit{CounterRead}$ operation $op$. Denote by $\sigma'$ the value of $H[p'].sn$ read by $p$ at line \ref{READ:firstsn} in the execution of $op$. According to line \ref{READ:help}, $\sigma-\sigma'\geq 2$. This means that process $p'$ executes line \ref{INC:incrementsn} at least twice during the execution interval of $op$. In particular $p'$ executes the step that set $switch_{val}$ to $1$ after $op$ was invoked by $p$. 
This proves the claim.
 \end{proof}

\subsubsection{Linearizability}
We next define the linearization $L$ of the operations in $E$ by first removing any $\mathit{CounterRead}$ operation that did not complete and any incomplete $\mathit{CounterIncrement}$  operation that has not successfully executed line \ref{INC:setswitch}.

Let $OP_W$ be the set of (complete and incomplete) $\mathit{CounterIncrement}$ operations that sucessfully set a switch while executing line \ref{INC:setswitch}. Let $OP_{LO}$ be the remaining complete $\mathit{CounterIncrement}$ operations in $E$ and $OP_R$ be the set of complete $\mathit{CounterRead}$ operations in $E$. Observe that each $\mathit{CounterIncrement}$ operation succesfully sets at most one switch, and each switch is succesfully set by at most one process. Thus we can univocally associate each operation in  $OP_W$ with the switch it sets.
We order the operations in $OP_W \cup OP_{LO} \cup OP_R$, according to the following rules :
\begin{enumerate}
\item We linearize each operation in $OP_W$ at the step where it sets its corresponding switch. By claim \ref{claim:incrementalSwitch}, operations in $OP_W$ are totally ordered and this order respect the real-time order. In the following we denote $opw_i$ the $i$-th operation in $OP_W$ according to our linearization order with $i\geq 0$.
\item  We linearize a $\mathit{CounterRead}$ operation $opr$ according to whether it returns normally or through the helping mechanism:
\begin{enumerate}
\item If $opr$ returns $\mathit{ReturnValue(p,q)}$ normally at line \ref{READ:normalreturn}, then it is linearized at the step where it reads the value $1$ of $switch_{q\cdot k +p}$ at line \ref{READ:whilecondition} \label{LINEARIZATION:normalreturn}. This is well-defined because this read primitive exists and it is unique (it is easy to check from the pseudo-code).
\item If $opr$ returns $\mathit{ReturnValue(p,q)}$ via the helping mechanism at line \ref{READ:helpingreturn}, then the operation is linearized immediately after $opw_{q\cdot k +p}$ \label{LINEARIZATION:helpingreturn}.
\end{enumerate}
\item Let L$_{WR}$ denote the linearization of all operations in $OP_W  \cup OP_R$ according to rule 1 and 2,
we linearize an operation $op$ in $OP_{LO}$ immediately before the first operation $op'$ in L$_{WR}$ that follows $op$ in the real-time order or at the end of L$_{WR}$ if $op'$ does not exist. 
\end{enumerate}
$\mathit{CounterRead}$ operations that returns $0$ after reading $switch_0=0$ are linearized before $opw_0$. If several operations are ordered at the same position, they are ordered respecting their real time order.

Linearization rule 2 and Lemma \ref{lemma:HelpingRead} implies the following claim.

\begin{claim} \label{claim:linearizationpoint}
Let $opr$ be a $\mathit{CounterRead}$ operation. We have that $opr$ is linearized at some point after its invocation.
\end{claim}


 \begin{lemma}[Linearizability]\label{lemma:unboundedcounter-linearizability}
 Algorithm \ref{algo:unboundedcounter} is a linearizable implementation of a $k$-multiplicative-accurate unbounded counter.
 \end{lemma}
  \begin{proof}
Let $op_1$ and $op_2$ be two operations in $E$ such as $op_1$ ends before $op_2$ is invoked. We prove that the linearization order $L$ respects the real-time order, thus $op_1$ precedes $op_2$ in $L$. First, we have the following claim:
\begin{itemize}
\item Let $op_1$ and $op_2$ be two $\mathit{CounterIncrement}$ operations. If at least one of these operations is in $OP_{LO}$, the claim trivially follows from rule 3. Otherwise it is already proved in rule 1. 
\item Let $op_1$ and $op_2$ be two $\mathit{CounterRead}$ operations. If both $op_1$ return normally the claim trivially holds from rule \ref{LINEARIZATION:normalreturn} and claim \ref{claim:linearizationpoint}.  So consider that $op_1$ returns through the helping mechanism and let $h_1=q\cdot k+p$ be the index of the switch read by $op_1$ at line \ref{READ:helpingvalue}, the last time before returning. According to rule  \ref{LINEARIZATION:helpingreturn}, $op_1$ is linearized immediately after $opw_{h_1}$. Also, by Lemma \ref{lemma:HelpingRead} and rule 2, $op_2$ is linearized after $opw_{h_1}$. The claim follows since according our linearization rules, If several operations are ordered at the same position, they are ordered respecting their real time order.
\item Consider that $op_1$ is a $\mathit{CounterIncrement}$ and $op_2$ is a $\mathit{CounterRead}$ operation. The claim follows from rules 1 and 2 and claim  \ref{claim:linearizationpoint} (the reverse follows a similar reasoning).
\end{itemize}

Next claim will be useful for proving that the ordering $L$ is consistent with the sequential specification of the $k$-multiplicative-accurate counter. 
\begin{claim}\label{claim:min-maxINC} Let $op$ be a $\mathit{CounterRead}$ operation invoked by a process $p_i$ that returns $\mathit{ReturnValue(p,q)}$. Then, the number of $\mathit{CounterIncrement}$ operations linearized before $op$ in $L$, denoted $v$, is at least $u_{min}=1+ \sum_{l=1}^{q}k^{l+1} + p\cdot k^{q+1}$ and at most $u_{max}=1+ \sum_{l=1}^{q}k^{l+1} +p(k-1)k^{q+1}+ n(k^{q+ 1}-1)$ where $n$ is the number of processes.
\end{claim}
\begin{proof}
Let $op$ be a $\mathit{CounterRead}$ operation invoked by a process $p_i$ that returns $\mathit{ReturnValue(p,q)}$ and let $h=q\cdot k +p$ with $p\geq 0$. Consider the  $\mathit{CounterIncrement}$ operation by $p_j$ that set to 1 the switch$_h$, denoted $opw_h$.

$op$ is linearized at the step where it reads $switch_h$ if it returns normally, or immediately after $opw_h$. Thus, from our linearization rules, the minimal number of $\mathit{CounterIncrement}$ operations that are linearized before $op$ includes each $opw_i$ in $OP_W$ with $0\leq i \leq h$, and every $\mathit{CounterIncrement}$ in $OP_{LO}$ linearized before $op$. 

We have by construction that each $switch_s$ in the $(l+1)$-th set of $k$ switches indexed in the interval $[l\cdot k +1 \dots (l+1)k]$ with $l\geq 0$, requires a process to perform $k^{l+1}$ $\mathit{CounterIncrement}$ operation instances before attempting to set $switch_s$ to $1$. In other words a process $p_i$ needs its local variable $lcounter_i$ to be equal to $k^{l+1}$ before it can attempt to set any $switch_s$ in $[l\cdot k +1 \dots (l+1)k]$ (line \ref{INC:condition-old}).  Since the value of $lcounter_i$ is reset to $0$ after a succesful $test\&set$ primitive is applied on a switch (line \ref{INC:resetCounter}), the sets of $\mathit{CounterIncrement}$ operation instances associated with any pair of succesful $test\&set$ primitives are disjoint.
Thus, $u_{min} = 1 + k\sum_{l=0}^{q-1}k^{l+1} + p\cdot k^{q+1}=1+ k\sum_{l=1}^{q}k^{l} + p\cdot k^{q+1}$ since we account for, in addition to the $p$ switches in the $(q+1)$-th set and $switch_0$, all $k$ switches in each of the sets indexed from $1$ to $q$.  

Similarily, we compute an upper bound $u_{max}$ on the maximum number of $\mathit{CounterIncrement}$ linearized before $op$. First suppose that $op$ returns normally. As already said, $op$ is linearized at the step where it reads $switch_h$ with $h=qk+p$. We have two possible cases either $p$ is equal to $0$ or it is equal to $1$ because the process checks the first and last switch of each set during the $\mathit{CounterRead()}$ instance. These two cases are depicted in Figure \ref{fig:umax} a) and b) respectively. If $p$ is equal to $0$, then process $p_i$ read $switch_{kq+1}=0$ in the execution of $op$, and according to our linearization rules $opw_{kq+1}$ is linearized after $op$. In a similar way, if $p$ is $1$ $p_i$ read $switch_{(q+1)k}=0$ and $opw_{(q+1)k}$ is linearized after $op$. However, in this second case, all the $k-1$ switches$_j$ with $j\in [q\cdot k +2 \dots (q+1)k-1]$ may have been set to $1$ before $op$ applied its read to switch$_{qk+1}$, and all the corresponding $opw_j$ may be linearized before $op$. Thus, the number of $opw$ linearized before $op$ is smaller than or equal to  $1+ \sum_{l=1}^{q}k^{l+1} +p(k-1)k^{q+1}$.
It remains to count the number of $\mathit{CounterIncrement}$ in $OP_{LO}$ linearized before $op$.
For every process $p_i$ the value of $lcounter_i$ is smaller than $k^{q+1}$ immediately before $p$ read either $switch_{kq+1}=0$ or $switch_{(q+1)k}=0$ in the execution of $op$. Since a process resets the value of its local counter only when it succeeds to set a switch to $1$ (line \ref{INC:resetCounter}), $lcounter_i$ defines the number of $\mathit{CounterIncrement}$ by $p_i$ in $OP_{LO}$ that are linearized before $op$.  Therefore, $u_{max}=1+ \sum_{l=1}^{q}k^{l+1} +p(k-1)k^{q+1}+ n(k^{q+1}-1)$ where $n$ is the number of processes.
If $op$ returns via the helping mechanism, then according to rule \ref{LINEARIZATION:helpingreturn}, it is linearized immediately after $opw_{q\cdot k +p}$ with $0 \leq p<k$. Thus, $1+ \sum_{l=1}^{q}k^{l+1} +pk^{q+1}$ is the number of $\mathit{CounterIncrement}$ in $OP_{W}$ linearized before $op$. Since $p<k$ the local counter of every process immediately after $opw_{q\cdot k +p}$ sets the corresponding switch is smaller than $k^{q+1}$. Since $k>1$, the claim follows.
\end{proof}
Let $op$ be a $\mathit{CounterRead}$ operation and let $v_{op} = \mathit{ReturnValue(p,q)}$ be the value it returns.
According to lines  \ref{READ:firstswitch}, \ref{READ:previoussets} and \ref{READ:ktimesret} of Algorithm \ref{algo:unboundedcounter}, $v_{op} = k(1 + \sum_{l=1}^{q}k^{l+1} + p\cdot k^{q+1})$; that is $v_{op} = k\cdot u_{min}$. 
According to claim \ref{claim:min-maxINC}, the number of $\mathit{CounterIncrement}$ operations linearized before $op$ in $L$, denoted $u$, is at least $u_{min}=1+ \sum_{l=1}^{q}k^{l+1} + p\cdot k^{q+1}$ and at most $u_{max}=1+ \sum_{l=1}^{q}k^{l+1} +p(k-1)k^{q+1}+ n(k^{q+ 1}-1)$ (where $n$ is the number of processes). And we have:
\begin{equation*}
\begin{aligned}
\frac{u_{max}}{k} &=\frac{1}{k} + \sum_{l=1}^{q}k^{l} + p\frac{k-1}{k}k^{q+1} + \frac{n}{k}(k^{q+1} -1)
\end{aligned}
\end{equation*}
\begin{equation*}
\begin{aligned}
\frac{u_{max}}{k} &\leq \sum_{l=1}^{q}k^{l} + p\cdot k^{q+1} + n\cdot k^q \\
\text{And } v_{op} &= k(1 + k\sum_{l=1}^{q}k^{l} + p\cdot k^{q+1}) \\
&= k(1 + k\sum_{l=1}^{q-1}k^{l} + k^{q+1} + p\cdot k^{q+1}) \\
&= k + k\sum_{l=2}^{q}k^{l} + p\cdot k^{q+2} + k^{q+2}
\end{aligned}
\end{equation*}
Thus, for $k \geq \sqrt n$, $\frac{u_{max}}{k} \leq v_{op}$. Therefore, since $p<k$, we have $\frac{u}{k} \leq \frac{u_{max}}{k} \leq v_{op} \leq k \cdot u_{min} \leq k \cdot u$. This completes the proof.
\end{proof}
\begin{figure}[hbtp] 
  \centering
\includegraphics[scale =0.4] {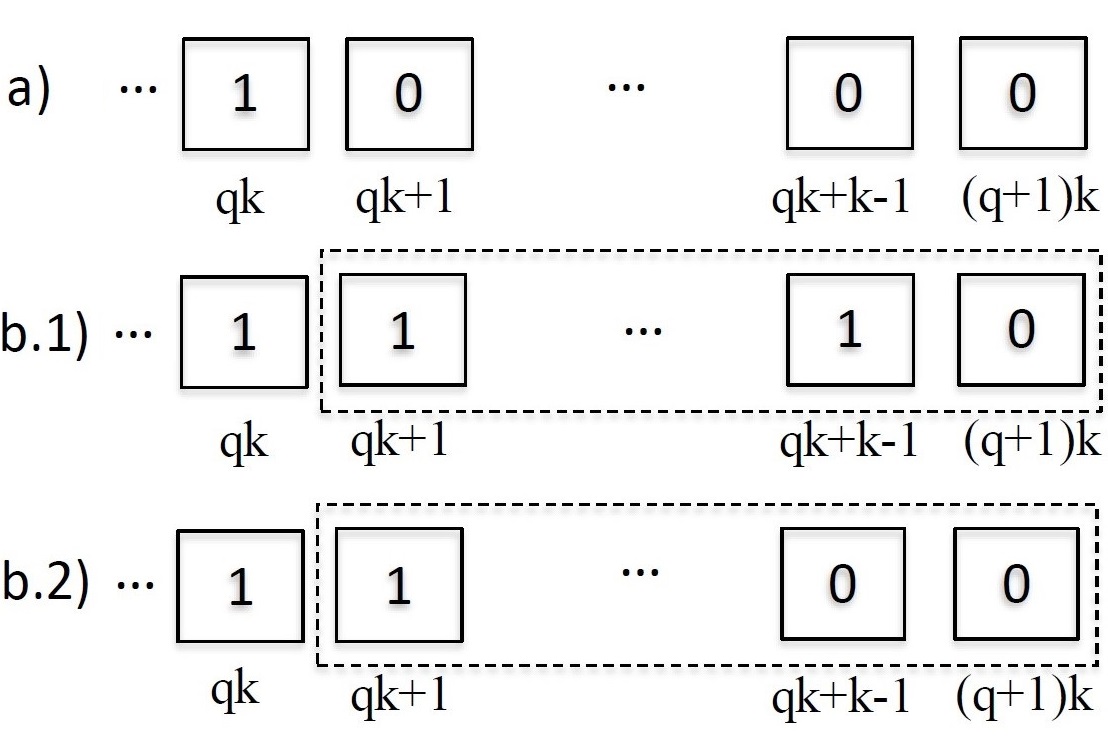} \caption{Switches state for the proof of claim \ref{claim:min-maxINC}. The dotted line indicate the $q+1$th interval of consecutive switches.When $p=1$, $op$ does not distinguish between cases b.1) and b.2)}
\label{fig:umax}
\end{figure}
\subsubsection{Complexity analysis}
\begin{lemma}\label{lemma:perprocInc} If process $p$ applies a $test\&set()$ primitive to a switch$_\alpha$ with $i \cdot k +1 \leq \alpha \leq (i+1) \cdot k$ for some integer $i \geq 0$, then $p$ has performed at least $k^{i+1}$ $\mathit{CounterIncrement()}$ operations. 
\end{lemma}
\begin{proof} Suppose that $p$ has executed a $test\&set()$ primitive to a switch$_\alpha$ with $i \cdot k +1 \leq \alpha \leq (i+1)\cdot k$  in the execution of a $\mathit{CounterIncrement()}$ operation $op$. According to line \ref{INC:forloop}, $j$ was equal to $i+1$ 
when computed at line \ref{INC:lcounter}, meaning that $\mathtt{lcounter}_p$ was equal to $k^{i+1}$. The claim holds because  $\mathtt{lcounter}_p$ is incremented only at line \ref{INC:incrementlcounter}, that is once for each $\mathit{CounterIncrement()}$ operation performed by $p$.
\end{proof}

\begin{lemma}[Amortized complexity]\label{lemma:unboundedcounter-complexity}
For $k \geq \sqrt{n}$, the amortized complexity of  Algorithm \ref{algo:unboundedcounter} is constant.
 \end{lemma}

 \begin{proof}
Let $E$ be a finite execution of the unbounded $k$-multiplicative-accurate counter object implemented in Algorithm \ref{algo:unboundedcounter}. Let $r$ denote the number of $\mathit{CounterRead()}$ instances in $E$ and $s$ be the number of $\mathit{CounterIncrement()}$ instances in $E$. We additionally denote $Ops_W(E)$ the set of $\mathit{CounterIncrement()}$ operations that execute at least one step in $E$, and $Ops_R(E)$ the set of $\mathit{CounterRead()}$ operations in $E$. We want to compute
\begin{equation*}
	\begin{aligned}
		AmtSteps(E) &= \frac{\sum_{ op \in Ops_W(E)\cup Ops_R(E)} Nsteps(op,E)}{r+s}
\end{aligned}	
\end{equation*}
where $Nsteps(op,E)$ is the number of steps executed by $op$ in $E$.

Let $Ops_{W_p}(E)$ denote the $\mathit{CounterIncrement()}$ operations in $Ops_{W}(E)$ executed by process $p$ and $s_p$ denote the total number of $\mathit{CounterIncrement()}$ operations executed by process $p$. Let $\alpha_p$ be the index of the furthest switch accessed by a process $p$ when executing any of the $\mathit{CounterIncrement()}$ operations in $Ops_{W_p}(E)$. We have that $i_p\cdot k +1 \leq \alpha_p \leq (i_p+1)\cdot k$ for some integer $i_p\geq 0$ (the case where $\alpha_p = 0$ is trivial). 

In the worst case, process $p$ applies a $test\&set()$ primitive to switch$_h$ for every $h \in [0,\ldots \alpha_p]$ and one additional step to write into $H[p]$ (line \ref{INC:writepair}) each time $p$ successfully set one of those switches. On the other hand, by Lemma \ref{lemma:perprocInc} if process $p$ applies a $test\&set()$ primitive to the switch$_{\alpha_p}$, then it has performed at least $k^{i_{p+1}}$ $\mathit{CounterIncrement()}$ operations. Therefore,
\begin{equation*}
\begin{aligned}
	\sum_{\substack{\displaystyle op \in Ops_{W_p}(E)}} Nsteps(op) &\leq 2 \cdot (i_p+1)k+1 \\
\text{And } 	s_p \geq k^{i_p+1}
\end{aligned}
\end{equation*}
\noindent  Thus, the total number of steps executed by the set of all processes $\mathcal{P}$ in order to perform the $\mathit{CounterIncrement()}$ operations in $E$ is : 
\begin{equation*}
\begin{aligned}
	\sum_{\substack{op \in Ops_W(E)}} Nsteps(op) &= \sum_{\substack{\displaystyle p \in \mathcal{P}}} \sum_{\substack{\displaystyle op \in Ops_{W_p}(E)}} Nsteps(op) \\
	& \leq \sum_{\substack{\displaystyle p \in \mathcal{P}}} 2 \cdot (i_p+1)k+1
\end{aligned}
\end{equation*}
	\begin{equation*}
\text{And }	s = \sum_{\substack{\displaystyle p \in \mathcal{P}}} s_p \geq \sum_{\substack{\displaystyle p \in \mathcal{P}}} k^{i_p+1}
\end{equation*}
\noindent Now we consider the number of steps applied by each process to perform $\mathit{CounterRead}$ operations. Let $\alpha$ be the index of the furthest switch set to $1$ by any process in $\mathcal{P}$. If $\alpha=0$ then the claim follows. Then suppose 
$i\cdot k +1 \leq \alpha \leq (i+1)\cdot k$ for some integer $i\geq 0$. For any sequence of switches with the index in $[j\cdot k +1,\ldots, (j+1)\cdot k]$ with $0\leq j \leq i$ a process $p$ only reads the first and the last switch in such interval (i.e.,  $switch_{j\cdot k +1}$ and $switch_{(j+1)\cdot k}$). This is because at the beginning $last_p$ is equal to $0$ and it is incremented by $1$ if it is a multiple of $k$ (at line \ref{READ:incrementlast}), by $k-1$ otherwise (line \ref{READ:incrementlast2}). Also, $last_p$ is a persistent variable, thus a process $p$ reads a given switch that has been set to $1$ at most once. This implies that the total number (in all its $\mathit{CounterRead}$ operations) of read primitives applied by a process $p$ to the switches is less or equal to $2(i+2)$ (2 per each of the $i+1$ intervals, plus $switch_0$ and $switch_{\alpha+1}$). 
Furthermore, any $\mathit{CounterRead()}$ operation executes $O(n)$ steps of the \textbf{for} loop at line \ref{forloop1} or line \ref{forloop2} once every $n$ iterations of the \textbf{while} loop (when the condition of line \ref{READ:conditionhelping} is satisfied). This means that the total number of steps executed by a process $p$ when performing its  $\mathit{CounterRead()}$ operations is less or equal to $4(i+2)$. Thus,
	\begin{equation*}
\begin{aligned}
	\sum_{\substack{\displaystyle op \in Ops_R(E)}} Nsteps(op) &\leq \sum_{\substack{\displaystyle p \in \mathcal{P}_r}} 4(i +2)
	& \leq 4(i +2)\cdot n_r
\end{aligned}
\end{equation*}

    \noindent where $\mathcal{P}_r$ is the set of processes who have invoked at least one $\mathit{CounterRead()}$ operation and $n_r$ is the cardinality of $\mathcal{P}_r$. Consider $n_r>0$, the other case is trivial.  Therefore:
		\begin{equation*}
		AmtSteps(E) \leq \frac{\sum_{\substack{\displaystyle p \in \mathcal{P}}} 2(i_p+1)k +1}{\sum_{\substack{\displaystyle p \in \mathcal{P}}} k^{i_p+1} + r} + \frac{4(i +2)\cdot n_r}{s + r}
		\end{equation*}		
	\noindent Furthermore, by lemma \ref{lemma:perprocInc} the minimum number of instances of the $\mathit{CounterIncrement()}$ operation executed to set the switch $\alpha$ is $k^{i+1}$. Thus,
		\begin{equation*}
		\begin{aligned}
		AmtSteps(E) & \leq \frac{\sum_{\substack{\displaystyle p \in \mathcal{P}}} 2(i_p+1)+ \frac{1}{k}}{\sum_{\substack{\displaystyle p \in \mathcal{P}}} k^{i_p} + \frac{r}{k}} + \frac{4(i +2)\cdot n_r}{k^{i +1}+ r}
		\end{aligned}
		\end{equation*}
		\noindent We have $k^x \geq  x+1 $ for $k\geq e$ and $\forall x \in \mathbf{R}$, it follows:
		\begin{equation*}\label{eq:totalsteps}
		\begin{aligned}
		\frac{\sum_{\substack{\displaystyle p \in \mathcal{P}}} 2(i_p+1)+ \frac{1}{k}}{\sum_{\substack{\displaystyle p \in \mathcal{P}}} k^{i_p } + \frac{r}{k}} &\leq \frac{\sum_{\substack{\displaystyle p \in \mathcal{P}}} 2(i_p+1)+ \frac{1}{k}}{\sum_{\substack{\displaystyle p \in \mathcal{P}}} (i_p +1)}
		\end{aligned}
		\end{equation*}
 \noindent If $i=0$, and since $r \geq n_r$ we have:	
		\begin{equation*}
		\begin{aligned}
		\frac{4(i+2)\cdot n_r}{k^{i +1}+ r} &\leq \frac{8\cdot n_r}{k+r} & \leq 8 
		\end{aligned}
		\end{equation*}
\noindent If  $i \geq 1$, because $n_r \leq n$ and $k^{i+1} \geq i\cdot k^2$ we have:
		\begin{equation*}
		\begin{aligned}
		\frac{4(i+2)\cdot n_r}{k^{i+1} + r} &\leq \frac{4(i+2)\cdot n}{i\cdot k^2 + r}
		\end{aligned}
		\end{equation*}

Resulting in an amortized complexity of $O(1)$ for $k \geq \sqrt n$.
 \end{proof}

From Lemma \ref{lemma:unboundedcounter-waitfreedom}, \ref{lemma:unboundedcounter-linearizability} and \ref{lemma:unboundedcounter-complexity} we conclude:
 \begin{theorem}\label{theorem:unboundedcounter}
 Algorithm \ref{algo:unboundedcounter} is a wait-free linearizable implementation of a $k$-multiplicative-accurate unbounded counter with a constant amortized complexity for $k \geq \sqrt{n}$.
 \end{theorem}

\newcommand{\pair}[1]{\mbox{$\langle #1\rangle$}}
\newcommand{\kCAS}{$k$-CAS}
\newcommand{\jCAS}{$j$-CAS}
\newcommand{\calC}{{\cal C}}

\subsection{An Amortized Step Complexity Lower bound for $k$-multiplicative accurate Counters}
\label{app:amortized-lb}
In this section, we prove that the amortized step complexity of solo-terminating
implementations of $k$-multiplicative accurate counters is $\Omega(\log_{q+1} \frac{n}{k^2})$ for $k \leq \sqrt{n/2}$, assuming the implementation uses base objects that support only read, write and either reading or regular conditional primitives of arity $q$ or less. A primitive has arity $q$ if it is applied atomically to a vector of $q$ base objects. In all the paper but this section, we consider $q=1$. Definitions and most of the technical lemmata hold from \cite{AttiyaH10}. In the following we provide the main idea of the lower bound and the Lemma that most differs from the original work. To this aim, in the following we remember some of the definitions formalized in \cite{AttiyaH10}.

Processes communicate with one another by issuing events that
apply \emph{read-modify-write} (RMW) primitives to vectors of base
objects. We assume that a primitive is always applied to vectors
of the same size and that all the base objects to
which a primitive is applied are over the same domain. 
A RMW primitive, applied to a vector of $k$ base objects over some
domain $D$, is characterized by a pair of functions, $\pair{g,h}$,
where $g$ is the primitive's \emph{update function} and $h$ is the
primitive's \emph{response function}. The update function $g: D^k
\times W \rightarrow D^k$, for some input-values domain $W$,
determines how the primitive updates the values of the base
objects to which it is applied.

Let $e$ be an event, issued by process $p$ after execution $E$, which applies the primitive
$\pair{g,h}$ to a vector of base objects $\pair{o_1,\ldots,o_k}$.
Then $e$ atomically does the following: it updates the values of
objects $o_1,\ldots,o_k$ to the values of the components of the
vector $g(\pair{v_1,\ldots,v_k},w)$, respectively, where
$\overrightarrow{v}=\pair{v_1,\ldots,v_k}$ is the vector of values
of the base objects after $E$, and $w \in W$ is an input parameter
to the primitive. We call $\overrightarrow{v}$ the
\emph{object-values vector} of $e$ after $E$. The RMW primitive
returns a response value, $h(\overrightarrow{v},w)$, to process
$p$. If $W$ is empty, we say that the primitive \emph{takes no
input}.


Next, we revise the concept of conditional synchronization primitives.
\begin{definition}
\label{def:conditionals}
A RMW primitive $\pair{g,h}$ is \emph{conditional} if,
for every possible input $w$,
$\Big{\vert} \big{\{} \overrightarrow{v} \vert g(\overrightarrow{v},w)
 \neq \overrightarrow{v} \big{\}} \Big{\vert} \leq 1$.
Let $e$ be an event that applies the primitive $\pair{g,h}$
with input $w$.
The \emph{change point} of $e$ is the unique vector $\overrightarrow{c_w}$
such that $g(\overrightarrow{c_w},w) \neq \overrightarrow{c_w}$;
any other vector is a \emph{fixed point} of $e$.
\end{definition}

For example, \kCAS{} is a conditional primitive for any integer $k \geq 1$. The single change point of a \kCAS{} event
with input $\pair{old_1,\ldots,old_k,new_1,\ldots,new_k}$ is the
vector $\pair{old_1,\ldots,old_k}$. Read is also a conditional
primitive, since read events have no change points.

%
%

A RMW event is invisible if its object-values vector is a fixed point of the event when it is issued.
A RMW event $e$ that is applied by process $p$ to an object vector is also invisible if, before $p$ applies another event, a write event is applied to each object $o_i$ that is changed by $e$ before another RMW event is applied to $o_i$.
If a RMW event $e$ is not invisible in an execution $E$ on some object $o$, we say that \emph{$e$ is visible in $E$ on $o$}. If $e$ is not invisible in $E$, we say that $e$ is a \emph{visible} event in $E$.

\paragraph*{Lower bound}
The key intuitions behind the following lower bound is that
first, in any $n$-process execution of a $k$-multiplicative accurate counter implementation,
`many' processes need to be aware of the participation of `many' other
processes in the execution, and second, if processes only use read,
write and conditional primitives, then a scheduling adversary can
order events so that information about the participation of
processes in the computation accumulates `slowly'. 

We formalize the first key intuition by proving Lemma \ref{newLemma5}. The rest is straightforward from the lower bound in the original work \cite{AttiyaH10}. 

The next definition captures the extent to which processes are aware of the
participation of other processes in an execution. Intuitively, a process $p$ is aware of the participation of another
process $q$ if $p$ reads a shared-memory value that was either directly
written by $q$ or indirectly influenced by a value written by $q$.
The following definitions formalize this notion.

\begin{definition}
\label{def:knows-about}
Let $e_q$ be an event by process $q$ in an execution $E$,
which applies a non-trivial primitive to a vector $v$ of base objects.
We say that an event $e_p$ in $E$ by process $p$ is \emph{aware} of $e_q$
if $e_p$ accesses a base object $o$
such that at least one of the following holds:
\begin{itemize}
\item
There is a prefix $E'$ of $E$ such that $e_q$ is visible on $o$ in $E'$
and $e_p$ is a RMW event that applies a primitive other than write to $o$,
and it follows $e_q$ in $E'$, or
\item there is an event $e_r$ that is aware of $e_q$ in $E$
and $e_p$ is aware of $e_r$ in $E$.
\end{itemize}

If an event $e_p$ of process $p$ is aware
of an event $e_q$ of process $q$ in $E$,
we say that $p$ is \emph{aware} of $e_q$
and that $e_p$ is \emph{aware} of $q$ in $E$.
\end{definition}

The following definition quantifies the extent to which a process
is aware of the participation of other processes in an execution.

\begin{definition}
\label{definition:process-familiarity} Process $p$ is \emph{aware}
of process $q$ after an execution $E$ if either $p=q$ or $p$ is
aware of an event of $q$ in $E$.
The \emph{awareness set} of $p$ after $E$, denoted $AW(E,p)$,
is the set of processes that $p$ is aware of after $E$.
\end{definition}

The following lemma proves a relation between the value returned by a
\emph{CounterRead} operation instance of a process in some execution
and the size of that process' awareness set after that execution. This is our main technical contribution to the lower bound.


\begin{lemma}\label{newLemma5}
Let $E$ be an execution of a solo-terminating $k$-multiplicative accurate counter object implementation where each process executes one instance of the $\mathit{CounterIncrement()}$ operation followed by one instance of the $\mathit{CounterRead()}$ operation. If the $\mathit{CounterRead()}$ instance by a process $p$ returns $i$ in $E$ then $|AW(E,p)| \geq \frac{i}{k}$.
\end{lemma}
\begin{proof}
Assume, by way of contradiction, that there is an execution $E$ where each process executes one instance of the $\mathit{CounterIncrement()}$ operation followed by one instance of the $\mathit{CounterRead()}$ operation, and a process $p$ such that a $\mathit{CounterRead()}$ instance by $p$, namely $op$, returns $i$ and $|AW(E,p)|< \frac{i}{k}$.

We construct a new execution $E'$ as follows: for any process $q \notin AW(E,p)$, we first remove all the events of $q$ from $E$; then, for any process $q'$, we remove all the events of $q'$ that are aware of $q$. Note that if an event $e_{q'}$ of $q'$ is aware of $q$, then all following events by $q'$ are also aware of $q$ and are removed. Also, no events of $p$ are removed since $p$ is aware only of processes in $AW(E,p)$.

We prove that $E'$ is an execution, and that it is indistinguishable from $E$. We consider events in the order they appear in $E'$. Let $e_q'$ be an event by process $q'$ that appears in $E'$, namely $E' =E'_1 e_q' E'_2$. Since $e_q'$ is also in $E$, we can also write $E =E_1 e_q' E_2$. For the induction, assume that $E'_1$ is an execution and that it is indistinguishable to every process that appears in it from $E_1$. In particular, $q'$ does not distinguish between $E'_1$ and $E_1$ and takes the same step after both of them. 
To see why $q'$  obtains the same response in $e_q'$ after $E'_1$ and after $E_1$, note that it can return a different response only if in $E$, $e_q'$ is aware of an event $e$ that was removed from $E_1$. This happens only if $e$ is aware of some process $q \notin AW(E,p)$, meaning that in $E$, $e_q'$ is also aware of $q$, contradicting the fact that $e_q'$ was not removed. Hence $E'_1 e_q'$ is an execution and $q'$ does not distinguish between $E'_1 e_q'$ and $E_1 e_q'$.

This implies that the $\mathit{CounterRead()}$ instance by $p$ returns $i$ also in $E'$; on the other hand, less than $\frac{i}{k}$ processes participate in $E'$. Let $E"$ be the extension of $E'$ in which the processes that participate in $E'$ complete their operation instances, one at a time. This execution exists by solo-termination, and results in a quiescent execution. However, less than $\frac{i}{k}$ instances of $\mathit{CounterIncrement()}$ operations completed in $E"$, and we have that $p$ returns $i$ when invoking $op$. Thus, the response of the $op$ is not linearizable. In particular, consider any linearization $\emph{L}$ of $E"$ and let $v$ be the number of  $\mathit{CounterIncrement()}$ instances linearized before $op$ in $\emph{L}$, we have that $\frac{v}{k} \leq i \leq k\cdot v < k\cdot \frac{i}{k} = i$.
\end{proof}

\begin{corollary}\label{corollary:sizeAwarenessSet}
Let $E$ be a quiescent $n$-process execution of a solo-terminating $k$-multiplicative counter implementation, where each process executes one instance of the $\mathit{CounterIncrement()}$ operation followed by one instance of a  $\mathit{CounterRead()}$ operation. Then,  the awareness sets of  $\frac{n}{2}$ processes contain at least $\frac{n}{2k^2}$ other processes after $E$.
\end{corollary}
\begin{proof}
Let $\emph{L}$ denote any linearization of $E$, and let $op$ be the $i$-th $\mathit{CounterRead()}$ instance in $\emph{L}$. Since $op$ is the $i$-th instance of $\mathit{CounterRead()}$ in $\emph{L}$, it returns $v$ such that $v \geq \frac{i}{k}$. By considering the last $\frac{n}{2}$ processes linearized and by Lemma \ref{newLemma5}, the claim follows.
\end{proof}

Information about processes that participate in an execution is transferred through base objects. To complete the proof, we need to show that each of the $\frac{n}{2}$ processes has to apply $\Omega (\log_{q+1}(n/k^2))$ events to build its awareness set (proved in the Appendix).

The following theorem formalize our step complexity lower bound.

\begin{theorem}\label{theorem:amortized-step-complexity}
Let $A$ be an $n$-process solo-terminating implementation of a $k$-multiplicative counter from base objects that support only read, write and either reading or regular conditional primitives of arity $q$ or less. Then $A$ has an execution $E$ that contains $\Omega (n \log_{q+1}(n/k^2))$ events for $k \leq \sqrt{n/2}$, in which every process performs a single $\mathit{CounterIncrement()}$ instance and a single $\mathit{CounterRead()}$ instance.
\end{theorem}

\section{A $k$-multiplicative-Accurate Bounded Max Register}
\label{sec:bounded-maxreg}

In this section, we present an implementation of a $k$-multiplicative-accurate max register. The algorithm is wait-free, and has asymptotically optimal worst-case step complexity. Indeed, we present in the next section a matching lower bound.

The key idea of our algorithm is to consider the $k$-base representation of values written to the register and have \texttt{Write} operations store only the index of the bit preceding (i.e., to the left of) the most significant bit (MSB) of their arguments. These indices are stored in an (accurate) $\big((\lfloor \log_k (m-1) \rfloor)+1\big)$-bounded max register implemented in a wait-free manner \cite{AspnesAC2012}. A $Read$ operation $R$ reads the value $p$ of the accurate max register. If it equals 0 (implying that it was not written to yet), $R$ returns 0. Otherwise, $p$ is the largest index written so far to the accurate max register and $R$ returns $k^{p}$. The pseudocode is presented by Algorithm \ref{algo:k-multiplicative-max-reg}.

\remove{
The main idea of the algorithm consists of representing the values written to the max register in a k base (i.e. $v = (b_{p-1}....b_1b_0)_k$ with $b_i \in \{0,...,(k-1)\}$). Consequently, the k-multiplicative approximation of a certain value $v$ is $k^{p}$ where $(p-1)$ is the most significant bit in the representation of $v$ in the $k$ base. Therefore, in Algorithm 1, writing a value is equivalent to storing the most significant bit $p$, and because the register is m-bounded, the maximum value $p$ can take is $log_k(m)$. Hence the need of a $log_k(m)$-bounded max register \textit{Max\_reg} to store $p$. The \texttt{Read()} instance of the max register object consists simply of reading the index $p$ stored in \textit{Max\_reg} and returning $k^{p+1}$ if $p$ is positive, or $0$ otherwise.
}

\begin{algorithm}[htbp]
  \SetAlgoLined
  \SetKwProg{Fn}{Function}{}{end}
  \textbf{Shared variables} $M$: $\big((\lfloor log_k(m-1) \rfloor)+1\big)$-bounded max register initially $0$\\
  \begin{multicols}{2}
\Fn{Read()}{
  $p \longleftarrow \text{M}.read()$ \label{MaxRead-MRead}\\
  \lIf{p=0} {return 0\label{MaxRead-return0}}
  \lElse{return $k^{p}$} }
{}
\columnbreak
\Fn{Write(v)}{$p \longleftarrow \lfloor log_k v \rfloor$ + 1; \label{MaxWrite-TakeLog}\\
  M.$write$($p$); \label{MaxWrite-MWrite}}{}
\end{multicols}
\caption{A $k$-multiplicative-accurate $m$-bounded max register}
\label{algo:k-multiplicative-max-reg}
\end{algorithm}

We now prove that Algorithm \ref{algo:k-multiplicative-max-reg} is a correct wait-free implementation of a $k$-multiplicative-accurate max register.

\begin{obs}\label{lemma:k-multiplicative-max-register-wait-freedom}
Algorithm \ref{algo:k-multiplicative-max-reg} is a wait-free implementation of a k-multiplicative-accurate m-bounded max register.
\end{obs}\label{Algo:boundedmaxreg}
\begin{proof}
Follows directly from the wait-freedom of the max register algorithm of \cite{AspnesAC2012}.
\end{proof}

\begin{lemma}\label{lemma:k-multiplicative-max-register-linearizability}
Algorithm \ref{algo:k-multiplicative-max-reg} is a linearizable implementation of a k-multiplicative-accurate m-bounded max register.
\end{lemma}

\begin{proof}
Let $\text{M}_m^k$ denote a $k$-multiplicative-accurate $m$-bounded max register implemented by Algorithm \ref{algo:k-multiplicative-max-reg} and let $E$ be an execution of $\text{M}_m^k$. We now specify how operation instances on $\text{M}_m^k$ in $E$ are linearized. First, all the instances of \texttt{Read} that did not execute line \ref{MaxRead-MRead} in $E$ and all the instances of \texttt{Write} operations did not execute line \ref{MaxWrite-MWrite} in $E$ do not appear in the linearization. We say these are \emph{removed operations}. Note that none of the removed operations has completed in $E$. For all remaining instances, we define the linearization point of a \texttt{Read} operation on $\text{M}_m^k$ to be the linearization point of the $read$ operation it invoked on M in $E$ (in line \ref{MaxRead-MRead}) and the linearization point of a \texttt{Write} operation on $\text{M}_m^k$ as the linearization point of the $write$ operation it invokes on M (in line \ref{MaxWrite-MWrite}). Since each non-removed operation instance on $\text{M}_m^k$ in $E$ is linearized at a step it performs (hence during its execution interval), the linearization order we have define, denoted by $L$, respects the real-time order of the operation instances in $E$.

It remains to show that $L$ satisfies the sequential specification of a k-multiplicative-accurate m-bounded max register. First note that since values written to $\text{M}_m^k$ are from $\{1,...,m-1\}$ and from lines \ref{MaxWrite-TakeLog}-\ref{MaxWrite-MWrite}, only values from $\{1,...,\lfloor \log_{k} (m-1) \rfloor +1\}$ are written to M. Let $R$ denote a \texttt{Read} instance in $L$ that returns $0$ in line \ref{MaxRead-return0}. Since only positive values are ever written to M, it follows that $R$ is not preceded in $L$ by any \texttt{Write} instance, hence the value of $\text{M}_m^k$ when $R$ is linearized is its initial value $0$, so $R$ returns the exact value of $\text{M}_m^k$.

Assume, then, that $R$ \emph{is} preceded in $L$ by one or more \texttt{Write} instances and returns a positive value $x=k^p$ for some $p \geq 1$. We need to prove that $v/k \leq x \leq vk$ holds, where $v$ is the maximum value written by any \texttt{Write()} instance linearized before $R$ in $L$. Since M is linearizable and since we have linearized all non-removed instances applied to $\text{M}_m^k$ in $E$ according the order of the operations they applied to M (in line \ref{MaxRead-MRead} or in line \ref{MaxWrite-MWrite}), there exists a \texttt{Write} operation that writes some value $w$ and appears before $R$ in $L$, such that $\lfloor \log_k w \rfloor = p-1$ and $p$ is the maximum value written to M by any \texttt{Write} instance that precedes $R$ in $L$. Let $V = \{w \big{\vert} \lfloor log_k(w) \rfloor = p-1\} $ be the set of all the values written to $\text{M}_m^k$ in $L$ before $R$ whose MSB equals $p-1$. Let $v = max(V)$. It follows that $v$ is the maximum value written to $\text{M}_m^k$ by any \texttt{Write()} instance linearized in $L$ before $R$. We have $v \in [k^{p-1} , k^{p}-1]$ and $x = k^{p}$. Consequently, $v \leq x \leq v \cdot k$ and the sequential specification of the k-multiplicative m-bounded max register is satisfied.
\end{proof}

\begin{theorem}\label{lemma:k-multiplicative-max-register-complexity}
Algorithm \ref{algo:k-multiplicative-max-reg} is a wait-free linearizable implementation of a k-multiplicative-accurate m-bounded max register with worst case operation step complexity $O\Big(min\big(\log_2(\log_k m),n\big)\Big)$.
\end{theorem}

\begin{proof}
Wait-freedom and linearizability follow from Observation \ref{lemma:k-multiplicative-max-register-wait-freedom} and Lemma \ref{lemma:k-multiplicative-max-register-linearizability}, respectively.
As for step complexity -- the worst case operation step complexity of the wait-free implementation of an $m$-bounded max register of \cite{AspnesAC2012} is $O\big(min(\log m,n)\big)$ for both $Read$ and $Write$ operations. Each operation of Algorithm \ref{algo:k-multiplicative-max-reg} applies a single operation on a $\big((\lfloor \log_k(m-1) \rfloor)+1\big)$-bounded max register and a constant number of additional steps. The theorem follows.
\end{proof}

\remove{
From Lemma \ref{lemma:k-multiplicative-max-register-wait-freedom}, \ref{lemma:k-multiplicative-max-register-linearizability} and \ref{lemma:k-multiplicative-max-register-complexity} we have
\begin{theorem}\label{theorem:k-multiplicative-max-reg}
Algorithm \ref{algo:k-multiplicative-max-reg} is a wait-free linearizable implementation of a k-multiplicative m-bounded max register with a worst case step complexity of $O(min(log_2(log_k(m)),n))$ for an instance of the operation \texttt{Read()} and $O(min(log_2(log_k(m)),n))$ for an instance of \texttt{Write($v$)}.
\end{theorem}
} 

\section{A Lower bound on the Worst Case Step Complexity of k-multiplicative m-bounded Max Register and Counter}

Aspnes et al. \cite{AspnesCAH16} proved a worst-case step complexity on the lower bound of a class of concurrent objects called \emph{L-perturbable}, that includes objects such as max registers, counters and snapshots. $L$ is called the \emph{perturbation bound}. Roughly speaking, an object is $L$-perturbable if, for every implementation of the object, there exists an operation $Op$ and an execution $E$, in the course of which $Op$ is ``perturbed'' $L$ times. An outstanding operation $Op$ by process $p$ is said to be perturbed by a process $q$, if a solo execution by $q$ can change the response of a solo execution by $p$. They prove \cite[Theorem 1]{AspnesCAH16} that any obstruction-free implementation of an $L$-perturbable object $O$ from \emph{historyless} primitives has an execution in which some process accesses $\Omega\big( min(\log_2 L, n) \big)$ distinct base objects during a single operation instance. Specifically, this implies that the worst-case step complexity of such implementations is $\Omega\big( min(\log_2 L, n) \big)$. 

\newtheorem*{def2}{[5], Definition 2}
\newtheorem*{def3}{[5], Definition 3}
\newtheorem*{theorem1}{[5], Theorem 1}

For the sake of presentation completeness, we restate the definition of an $L$-perturbable object from \cite{AspnesCAH16}.

\begin{def2}
\label{definition:L-perturbable-executions}
Let $\cal{I}$ be an obstruction-free implementation of an object. The set $S_k$ of $k$-perturbing executions with respect to an operation instance $op_n$ by process $p_n$ is defined inductively as follows:

\begin{enumerate}
\item $S_0$ is the singleton set containing the empty sequence.

\item  If $\alpha_{k-1} \lambda_{k-1}$ is in $S_{k-1}$, where $\lambda_{k-1}$ consists of $n-1$ events, one by each of the processes $p_1, \ldots, p_{n-1}$, then $\alpha_{k-1} \lambda_{k-1}$ is in $S_k$. In this case, we say that $\alpha_{k-1} \lambda_{k-1}$ is \emph{saturated}.

\item  Suppose $\alpha_{k-1} \lambda_{k-1}$ is in $S_{k-1}$, no process has more than one event in $\lambda_{k-1}$, and there is a sequence $\gamma$ of events by a process $p_l$ different from $p_n$ and the processes that have events in $\lambda_{k-1}$, such that the sequences of events by $p_n$ as it performs $op_n$ after $\alpha_{k-1} \lambda_{k-1}$ and $\alpha_{k-1} \gamma \lambda_{k-1}$ differ. Let $\gamma=\gamma' e \gamma''$, where $e$ is the first event of $\gamma$ such that the sequences of events taken by $p_n$ as it performs $op_n$ by itself after $\alpha_{k-1} \lambda_{k-1} $ and after $\alpha_{k-1} \gamma' e \lambda_{k-1}$ differ. Let $\lambda$ be some permutation of the event $e$ together with the events in $\lambda_{k-1}$, and let $\lambda'$, $\lambda''$ be any two sequences of events such that $\lambda=\lambda' \lambda''$. Then the execution $\alpha_k \lambda_k$ is in $S_k$, where $\alpha_k =\alpha_{k-1} \gamma' \lambda'$ and $\lambda_k=\lambda''$.
\end{enumerate}
\end{def2}

\begin{def3}
\label{definition:L-perturbable-objects}

An obstruction-free implementation of an object is $L$-perturbable if there is an operation instance $op_n$ such that the set $S_L$ of $L$-perturbing executions with respect to $op_n$ by $p_n$ is nonempty.
\end{def3}

An object $\cal{O}$ is \emph{perturbable} if all its obstruction-free implementations are perturbable.

\begin{theorem1}\label{theorem:step-lower-bound}
Let $A$ be an $n$-process obstruction-free implementation of an
$L$-perturbable object $\cal{O}$ from historyless primitives.
Then $A$ has an execution in which some process accesses
$\Omega(\min(\log_2{L},n))$ distinct base objects during 
a single operation instance.
\end{theorem1}


\begin{lemma}\label{proof:perturbableboundmaxreg}
A k-multiplicative-accurate m-bounded max register is 
$\Theta(log_k m)$-perturbable for $k>1$.
\end{lemma}

\begin{proof}
Let \textit{O} be a $k$-multiplicative-accurate $m$-bounded max register and consider an obstruction-free implementation of \textit{O}. We show that \textit{O} is $(\frac{1}{2} log_k(m-1))$-perturbable for a \texttt{Read()} operation instance $op_n$ by process $p_n$. We proceed by induction where the base case for $r=0$ is immediate.
Let $r < \frac{1}{2}log_k(m-1)$ and let $\alpha_{r-1} \lambda_{r-1}$ be an $(r-1)$-perturbing execution of \textit{O}. If $\alpha_{r-1} \lambda_{r-1}$ is saturated, then it is also an $r$-perturbing execution. Otherwise, denote by $v_{r-1}$ the maximum input to the \texttt{write()} operations linearized before $op_n$ in the execution sequence $\alpha_{r-1} \lambda_{r-1}$. Since $\alpha_{r-1} \lambda_{r-1}$ is not saturated, there exists a process $p_l \neq p_n$ that does not take steps in $\lambda_{r-1}$. Let $\gamma$ be the execution fragment by $p_l$ where it finishes any incomplete operation in $\alpha$ and then performs a \texttt{write()} operation to the max register with the value $v_{r} = k^2v_{r-1} +1$. Then $op_n$ must return a value $x$ such that $kv_{r-1} < v_{r}/k \leq x \leq kv_{r}$ when run after $\alpha_{r-1} \gamma \lambda_{r-1}$ . It follows that an $r$-perturbing execution can be constructed from $\alpha_{r-1} \lambda_{r-1}$ and $\gamma$ as specified by [5], Definition 2. Because \textit{O} is an m-bounded max register, during the $r$th step of the induction, the value written to the max register must satisfy $v_{r} \leq m -1$. Consequently it suffices to have:
\begin{equation*}
v_{r} \leq (k+1)^{2r} \leq m -1 \\
\implies r \leq \frac{1}{2}log_{k+1}(m -1)= \Theta(log_{k} m)
\end{equation*}
\end{proof}

from Lemma \ref{proof:perturbableboundmaxreg} and [5], Theorem 1 we have the following theorem:
\begin{theorem}
\label{theorem:approx-max-lower-bound}
The worst-case step complexity of a $k$-multiplicative $m$-bounded max register is $\Omega\big( min(\log_2 (\log_k m), n) \big)$
\end{theorem}

\begin{lemma}\label{proof:perturbableboundedcounter}
A k-multiplicative-accurate m-bounded counter is $\Theta(log_k(m))$-perturbable for $k>1$.
\end{lemma}

\begin{proof}
Let \textit{O} be a $k$-multiplicative $m$-bounded counter and consider an obstruction-free implementation of \textit{O}. We show that \textit{O} is $(\frac{1}{2} log_k(m-1))$-perturbable for a $\mathit{CounterRead()}$ operation instance $op_n$ by the process $p_n$. We proceed by induction where the base case for $r=0$ is immediate.
Let $\alpha_{r-1} \lambda_{r-1}$ be an $(r-1)$-perturbing execution of \textit{O}. If $\alpha_{r-1} \lambda_{r-1}$ is saturated, then it is also an $r$-perturbing execution. Otherwise, let $I_r$ denote the number of $\mathit{CounterIncrement()}$ operation instances performed by the perturbing process in iteration $r$. We have that $I_1 = 1$ in order for $op_n$ to return a value greater than $0$. For $r>1$, if $op_n$ runs after $a_{r-1} \lambda_{r-1}$ it can return a value that is as large as $k \cdot \sum_{j=1} ^{r-1} I_j$. Therefore, we need the number of complete $\mathit{CounterIncrement()}$ operation instances after $a_{r-1} \gamma \lambda_{r-1}$ to be at least $k^2 \cdot \sum_{j=1} ^{r-1} I_j +1$ for $op_n$ to return a value greater than $k \cdot \sum_{j=1} ^{r-1} I_j$.  

Besides the $\mathit{CounterIncrement()}$ operation instances in $\gamma$, at least  $\sum_{j=1} ^{r-1} I_j - (r-1)$ have finished, therefore setting $I_r= (k^2-1) \cdot \sum_{j=1} ^{r-1} I_j +r$ implies that $op_n$ returns at least 
$\frac{1}{k} (\sum_{j=1} ^{r-1} I_j -(r-1)+I_r)=\frac{1}{k} (\sum_{j=1} ^{r-1} I_j - (r-1)+(k^2-1) \cdot \sum_{j=1} ^{r-1} I_j +r)=\frac{1}{k} (k^2\cdot \sum_{j=1} ^{r-1} I_j +1)$ which is greater tha $k \cdot \sum_{j=1} ^{r-1} I_j$ as needed.

\begin{equation*}
\begin{aligned}
I_r &= \sum_{i=0} ^{r-1}(r-i)(k^2-1)^i =\sum_{i=1} ^{r}i \cdot (k^2-1)^{r-i} \\
&=(k^2-1)^r \sum_{i=1} ^{r} \frac{i}{(k^2-1)^{i}} \\
&=\frac{(k^2-1)((k^2-1)^r -1)+r(2-k^2)}{(k^2-2)^{2}}\leq k^{2r} \leq m \\ 
& \implies r \leq \frac{1}{2}log_{k}(m)= \Theta(log_{k} m)
\end{aligned}
\end{equation*}
\end{proof}

From Lemma \ref{proof:perturbableboundedcounter} and [5], Theorem 1, we prove the following Theorem

\begin{theorem}
\label{theorem:approx-counter-lower-bound}
The worst-case step complexity of a $k$-multiplicative $m$-bounded counter is $\Omega\big( min(\log_2 (log_k m), n) \big)$
\end{theorem}



\section{Discussion}
We have presented upper and lower bounds on the step complexity of a variant of deterministic approximate counters and max-registers. 
We have proved the possibly counter-intuitive\footnote{Pun unintended.} result that when the accuracy parameter $k$ does not depend on $n$, relaxing counter semantics by allowing inaccuracy of a multiplicative factor cannot asymptotically reduce the step complexity of unbounded counters by more than a logarithmic factor. Then, we present a wait-free linearizable $k$-multiplicative-accurate counter for $k\geq \sqrt{n}$ with constant amortized step complexity. 

The maximum improvement in the worst case step complexity of the bounded variant of $k$-multiplicative-accurate counters remains an open question. Also, when $k$ is constant, it is unclear whether there exists a deterministic wait-free $k$-multiplicative-accurate counter implementation with $o(\log^2 n)$ amortized step complexity.

We also show that relaxing the semantics of max registers by allowing inaccuracy of even a constant multiplicative factor yields an exponential improvement in the worst-case step complexity of the bounded variant and in the amortized step complexity of the unbounded one.  
Overall we provide theoretical evidence that worst-case time complexity does not indicate benefit from some common relaxation of unbounded counters and max-register while average complexity does. It is interesting to note that a similar result has been proved for relaxed queues in the message passing systems by Talmage and Welch  \cite{TalmageW14}.

\section*{Acknowledgment}
Alessia Milani and Corentin Travers are supported by ANR projects Descartes and FREDDA. Adnane Khattabi is supported by UMI Relax. Danny Hendler is supported in part by ISF grant 380/18.



%

\end{document}